\documentclass{llncs}
\usepackage{makeidx}
\usepackage{amsfonts}
\usepackage{amsmath}
\usepackage{a4wide}
\usepackage{graphicx}
\usepackage{algorithm}
\usepackage{algorithmic}
\usepackage[USenglish]{babel}

\newtheorem{observation}{Observation}

\newif\ifcomment\commentfalse

\def\commentOFF{\commentfalse}

\long\outer\def\bc#1\ec{{\ifcomment \sloppy  $[${\bf Leen suggests}]
{{#1}} \textbf{[end]} \fi }}

\long\outer\def\br#1\er{{\ifcomment \sloppy  $[${\bf Leen suggests to remove}]
{{#1}} \textbf{[end]} \fi }}

\long\outer\def\bo#1\eo{{\ifcomment \sloppy  $[${\bf instead of}]
{\textit{#1}} \textbf{[end]}  \fi }}

\long\outer\def\BC#1\EC{{\ifcomment \sloppy \par \#  \dotfill
{\textsc{#1}} \dotfill \# \par \fi }}

\commentOFF

\begin{document}

\title{Constructing level-2 phylogenetic networks from triplets\thanks{Part of this research has been funded
by the Dutch BSIK/BRICKS
project.}}
\titlerunning{Constructing level-2 phylogenetic networks}  
%
\author{Leo van Iersel\inst{1}, Judith Keijsper\inst{1}, Steven
Kelk\inst{2}, Leen Stougie\inst{2}}

\authorrunning{Van Iersel, Keijsper, Kelk, Stougie}

\institute{Technische Universiteit Eindhoven (TU/e), Den Dolech 2, 5612 AX Eindhoven, Netherlands.\\
\email{l.j.j.v.iersel@tue.nl, j.c.m.keijsper@tue.nl}\\
\and
Centrum voor Wiskunde en Informatica (CWI), Kruislaan 413, 1098 SJ Amsterdam, Netherlands. \\
\email{S.M.Kelk@cwi.nl, Leen.Stougie@cwi.nl} \\
}

\maketitle              

\begin{abstract}
Jansson and Sung showed in \cite{JS1} that, given a dense set of input triplets $T$ (representing hypotheses about the
local evolutionary relationships of triplets of species), it is possible to determine in polynomial time whether there
exists a \emph{level-1 network} consistent with $T$, and if so to construct such a network. They also showed that,
unlike in the case of trees (i.e. level-0 networks), the problem becomes NP-hard when the input is non-dense. Here we
further extend this work by showing that, when the set of input triplets is dense, the problem is even polynomial-time
solvable for the construction of \emph{level-2} networks. This shows that, assuming density, it is tractable to
construct plausible evolutionary histories from input triplets even when such histories are heavily non-tree like.
This further strengthens the case for the use of triplet-based methods in the construction of phylogenetic networks.
We also show that, in the non-dense case, the level-2 problem remains NP-hard.
\end{abstract}

\section{Introduction}

\subsection{Phylogenetic reconstruction: popular methods}

Broadly speaking \emph{phylogenetics} is the field at the interface of biology, mathematics and computer-science which
tackles the problem of (re-)constructing plausible evolutionary scenarios when confronted with incomplete and/or
error-prone biological data. There are already a great many algorithmic strategies for constructing evolutionary
scenarios. The most well-known techniques are Maximum Parsimony (MP), Maximum Likelihood (ML), Bayesian methods,
Distance-based methods (such as Neighbour Joining and UPMGA) and Quartet-based methods, as well as various
(meta-)combinations of these. See
\cite{bryant01constructing}\cite{traditional}\cite{jiang00polynomial}\cite{phylogenetics} for good discussions of
these methods.\\
\\
The methods generally considered accurate enough to cope with large input data sets are MP and ML \cite{largescale},
with Bayesian methods (based on Markov Chain Monte Carlo random walks) more recently also emerging as a popular method
within molecular studies \cite{traditional}\cite{elephant}. However, MP and (especially) ML both suffer from slow
running times which means that finding optimal MP/ML solutions on data sets consisting of more than several tens of
species is practically infeasible. (Both problems are NP-hard \cite{roch}.) One response to this tractability problem
has been the development of Quartet-based methods. Such methods actually encompass an array of algorithms (e.g.
Maximum Quartet Consistency, Minimum Quartet Inconsistency) and various heuristics for rejecting problematic parts of
the input data (e.g. Q*/Naive Method, Quartet Cleaning and Quartet Puzzling.) The unifying idea however is the
assumption that, with high-accuracy, one can construct evolutionary trees for all, or at least very many subsets of
exactly 4 species. Given such ``quartets'' we then wish to find a single tree, containing all the species encountered
in the quartets, which is consistent with all - or at least, as many as possible - of the given quartets.

\subsection{From quartet methods to triplet methods}

Quartet methods apply to the construction of unrooted evolutionary trees; less well studied is the problem of
constructing \emph{rooted} evolutionary trees, where the edges of the tree are directed to reflect the direction of
evolution. (In unrooted evolutionary trees a path between two species A and B does not indicate whether A evolved into
B, or vice-versa.) The analogue of quartet methods in the case of rooted evolutionary trees are \emph{triplet}
methods: here we are given not unrooted trees on 4 leaves, but rooted binary trees on 3 leaves, see Figure
\ref{fig:singletriplet}. One can interpret the triplet in this figure as saying that species $x$ and $y$ only diverged
from each other \emph{after} some common ancestor of theirs had already diverged from species $z$. For any set of 3
leaves there are at most 3 triplets possible. There are various ways to generate triplets from biological data; a
high-accuracy method such as MP or ML is often used because for the construction of small trees their running time is
perfectly acceptable.

\begin{figure}
\centering \vspace{-0.5cm} \includegraphics{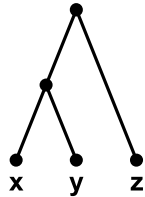} \caption{One of the three possible triplets on the set
of leaves $x, y, z$. Note that, as with all figures in this article, all arcs are assumed to be directed downwards,
away from the root.} \vspace{-0.3cm} \label{fig:singletriplet}
\end{figure}

\noindent Aho et al. studied the problem of constructing trees from input sets of triplets. They showed that, given a
set of triplets, it is possible to determine in polynomial time whether there exists a single rooted tree consistent
with all the input triplets \cite{aho}. (And, if so, to construct such a tree.) This contrasts favourably with the
corresponding quartet problem, which is NP-hard \cite{steelhard}. Various authors
\cite{bryant97}\cite{Gasieniec99}\cite{jansson01}\cite{jansson05}\cite{wu04} have studied the problem of, when
confronted with a set of triplets for which the Aho et al. algorithm fails to return a tree, finding a tree which is
best possible under some given optimisation criteria. A well-studied, albeit NP-hard, optimisation criteria is to find
a tree that maximises the number of input triplets it is consistent with.

\subsection{From trees to networks}

In recent years attention has turned towards the construction of evolutionary scenarios that are not tree-like. This
has been motivated by the fact that biological phenomena such as hybridisation, horizontal gene transfer,
recombination, and gene duplication can cause lineages which earlier in time diversified from a common ancestor, to
once again intersect with each other later in time. In essence, thus, evolutionary scenarios where the underlying,
undirected graph potentially contains cycles. Rather than attempt to summarise this extremely varied area we refer the
reader to \cite{husonbryant}, \cite{makarenkov} and \cite{moret2004}, all outstanding survey articles.\\
\\
Note that in the case named above the need for structures more general than trees is \emph{explicitly} motivated by
the \emph{inherently} non tree-like (plausible) evolution(s) that we are trying to reconstruct. However, even if
underlying evolutionary scenarios are believed to be tree-like, it can be extremely useful to have algorithms for
constructing more complex structures. This is because \emph{errors} in the input data (in this case, input triplets)
can create a situation where it is not possible to build a tree i.e. where the algorithm of Aho et al. fails. If
however a more complex, non tree-like evolutionary scenario \emph{can} be built from that triplet set, then this can
often be used to (visually) locate the parts of the input that are responsible for spoiling the expected tree-like
status of the input. This is indeed part of the motivation behind the well-known \emph{SplitsTree} package of Huson et
al. \cite{splitstree}. In the same way there are both explicit and implicit interpretations possible of the phenomenon
where, for some sets of three species, there is more than one triplet in the input. (Note that for any three species
there are at most three different triplets possible.) On one hand this can be viewed as a reflection of the fact that
the three species in question genuinely came from an evolutionary scenario that was non tree-like, and as such that
the multiple triplets corresponding to those three species are indeed all potentially valid. On the other hand one can
view multiple input triplets on the same set of three species as an expression of uncertainty/confidence as to which
triplet is the ``correct'' one. Suffice to say: in this paper we take a purely mechanical, algorithmic approach to
this question and leave it to the reader to reason about the relative merits of implicit and explicit interpretations.

\subsection{Efficiently constructing phylogenetic networks from dense sets of input triplets}

In \cite{JS2} and \cite{JS1} Jansson and Sung considered the following problem. Given a set of input triplets, is it
possible to construct a \emph{level-1} network (otherwise known as a \emph{galled tree} or a \emph{galled network})
which is consistent with all those triplets? Informally, a level-$k$ network (for $k\geq 0$) is an evolutionary
network where each biconnected component of the network contains at most $k$ recombination events. They showed that,
in general, the level-1 problem is NP-hard. (In contrast the algorithm of Aho et al. always runs in polynomial time.)
However, when the input is \emph{dense} - each set of 3 species has at least one triplet in the input - they show that
the problem can be solved in polynomial time. (In \cite{JS1} an algorithm is given with quadratic running time in the
number of input triplets, in \cite{JS2} this is improved to linear time.) Density is a reasonable assumption if
high-quality triplets can be constructed for all subsets of 3 species. In \cite{JS2} various upper-bounds,
lower-bounds and approximation algorithms for the general case are also given. (A similar group of authors has also
explored related problems of constructing galled trees from ultrametric distance matrices \cite{ultrametric}, and
building galled trees where certain input triplets are \emph{forbidden} \cite{forbid}.)\\
\\
In this paper we extend considerably the work of Jansson and Sung in \cite{JS2} by showing that, when the input set is
dense, it is even polynomial-time solvable to detect whether a \emph{level-2} network can be constructed consistent
with the input triplets. (And, if so, to construct one.) We give an algorithm that runs in time $O(|T|^3)$ where $T$
is the set of input triplets. This significantly extends the power of the triplet method because it further extends
the complexity of the evolutionary scenarios that can be constructed. For example, networks of the complexity shown in
Figure~\ref{fig:biglevel2} can be constructed by our algorithm. This tractability result follows from several deep
insights regarding the behaviour of the \emph{SN-set}, first introduced by Jansson and Sung, and by the construction
of algorithms for solving the \emph{simple level-2} problem. On the basis of this result it is tempting to conjecture
that, for fixed $k$, the dense level-$k$ problem is polynomial-time solvable. However it is not yet clear that the
pivotal theorem in Section \ref{sec:poly}, Theorem \ref{thm:one}, generalises easily to level-3 networks and higher;
already in level-2 it is no longer necessarily true that an SN-set corresponds to a single ``side network''. We also
show that the level-2 problem is NP-hard in the general case; this result also touches on some interesting issues
concerning the conditions under which triplets give rise to one, and one only, possible solution.

\begin{figure}
\centering \vspace{-0.2cm}
\includegraphics{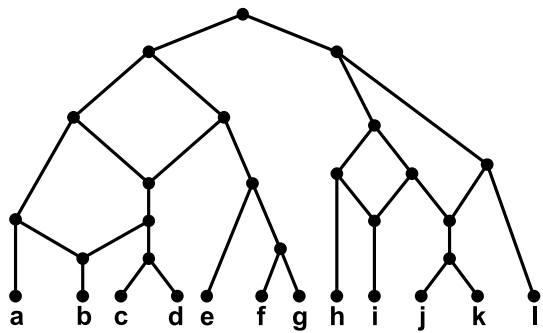}
\caption{An example of a level-2 network.} \label{fig:biglevel2}
\end{figure}

\vspace{.3cm}

\subsection{Layout of the paper}

In Section~\ref{sec:def} we introduce the basic definitions and notation used throughout this paper; more
context-specific terminology is introduced at relevant points throughout the paper. In Section~\ref{sec:poly} we
present the main result of this paper, the algorithm for constructing level-2 networks from dense triplet sets. The
result is rather complicated and for this reason we split this section into several sub-sections.
Section~\ref{sec:nphard} shows that, for general input sets, constructing level-2 networks is NP-hard. Finally in
Section~\ref{sec:conclusion} we discuss both our conclusions and the very many fascinating open problems which still
remain in this area. The appendix contains various proofs that would otherwise interrupt the flow of the paper.

\section{Definitions}
\label{sec:def}

A \emph{phylogenetic tree} is a rooted binary tree with directed edges (arcs) and distinctly labelled leaves. A
\emph{phylogenetic network} (\emph{network} for short) is a generalisation of a phylogenetic tree, defined as a
directed acyclic graph in which exactly one vertex has indegree 0 and outdegree 2 (the root) and all other vertices
have either indegree 1 and outdegree 2 (\emph{split vertices}), indegree 2 and outdegree 1 (recombination vertices) or
indegree 1 and outdegree 0 (leaves), where the leaves are distinctly labelled. A leaf that is a child of a
recombination vertex is called a \emph{recombination leaf}. In general directed acyclic graphs a \emph{recombination
vertex} is a vertex with indegree 2. A directed acyclic graph is \emph{connected} (also called ``weakly connected'')
if there is an undirected path between any two vertices and \emph{biconnected} if it contains no vertex whose removal
disconnects the graph. A \emph{biconnected component} of a network is a maximal biconnected subgraph and is called
\emph{trivial} if it is equal to two vertices connected by an arc. To avoid ``redundant'' networks, we assume in this
paper that in any network every nontrivial biconnected component has at least three outgoing arcs.
\begin{definition}
A network is said to be a \emph{level-$k$ network} if each biconnected component contains at most $k$ recombination
vertices.
\end{definition}
A network that is also a tree can thus be considered a level-0 network. A network that is a level-$k$ network but not
a level-$(k-1)$ network is called a \emph{strict level-$k$ network}.\\
\\
A phylogenetic tree with exactly three leaves is called a \emph{rooted triplet} (\emph{triplet} for short). The unique
triplet on a leaf set $\{x,y,z\}$ in which the lowest common ancestor of $x$ and $y$ is a proper descendant of the
lowest common ancestor of $x$ and $z$ is
denoted by $xy|z$ (which is identical to $yx|z$). The triplet in Figure \ref{fig:singletriplet} is $xy|z$.\\
\\
Denote the set of leaves in a network $N$ by $L(N)$ and for any set $T$ of triplets define $L(T) = \bigcup_{t \in T}
L(t)$. A set $T$ of rooted triplets is called \emph{dense} if for each $\{x,y,z\} \subseteq L(T)$ at least one of
$xy|z$, $xz|y$ and $yz|x$ belongs to $T$. Furthermore, for a set of triplets $T$ and a set of leaves $L'\subseteq L$,
we denote by $T|L'$ the subset of the triplets in $T$ that have only leaves in $L'$. The number of leaves $|L(N)|$ of
a network $N$ is denoted by $n$.

\begin{definition}
\label{def:con} A triplet $xy|z$ is \emph{consistent} with a network $N$ (interchangeably: $N$ is consistent with
$xy|z$) if $N$ contains a subdivision of $xy|z$, i.e. if $N$ contains vertices $u \neq v$ and pairwise internally
vertex-disjoint paths $u \rightarrow x$, $u \rightarrow y$, $v \rightarrow u$ and $v \rightarrow z$.
\end{definition}
By extension, a set of triplets $T$ is said to be \emph{consistent} with $N$ (interchangeably: $N$ is consistent with
$T$) if every triplet in $T$ is consistent with $N$. We say that a set of triplets $T$ is \emph{level-$k$ realisable}
if there exists a level-$k$ network $N$ consistent with $T$. To clarify triplet consistency we observe that the
network in Figure \ref{fig:biglevel2} is consistent with (amongst others) $ab|c$, $bc|a$ and $dg|k$ but not consistent
with (for example) $ah|f$ or $hk|i$.\\

\pagebreak

\noindent Note that the definition of triplet consistency in \cite{JS2} (``$xy|z$ is \emph{consistent} with $N$ if
$xy|z$ is an embedded subtree of $N$ (i.e. if a lowest common ancestor of $x$ and $y$ is a proper descendant of a
lowest common ancestor of $x$ and $z$)'') is only usable for trees and not for general networks. Personal
communication with the authors \cite{communicate} has clarified that Definition~\ref{def:con} is the definition they
actually meant. It follows directly from Definition~\ref{def:con} that triplet consistency can be checked in
polynomial time (in the number of vertices), since a fixed number of disjoint paths in a directed acyclic graph can be
found in polynomial time \cite{fortune}. Note that the algorithms in Section \ref{sec:poly} only need to check triplet
consistency in a very restricted type of networks, making it possible to check $O(n^3)$ triplets in $O(n^3)$ time.
Designing a fast algorithm (faster than searching for disjoint paths) that checks triplet consistency in general
networks is an open
problem.\\
\\
We will now define SN-sets, introduced in \cite{JS1}, which will play an important role in the rest of the paper. For
a triplet set $T$, let $\mathcal{S}_T$ be the operation on subsets $X$ of $L(T)$ defined by $\mathcal{S}_T(X) = X \cup
\{c \in L(T) | \exists x,y \in X: xc|y \in T\}$. The set $SN_T(X)$ is defined as the closure of $X$ w.r.t. the
operation $\mathcal{S}_T$. Define an \emph{$SN$-set} of $T$ as a set of the form $SN_T(X)$ for some $X\subseteq L(T)$,
i.e. SN-sets are the subsets of $L(T)$ that are closed under the operation $\mathcal{S}_T$. Note that
$SN_T(L(T))=L(T)$ so $L(T)$ is an SN-set. Note also that $SN_T(\{x\}) = \{x\}$ for each $x \in X$; we call such an
SN-set a \emph{singleton} SN-set. An SN-set $X$ is \emph{maximal} with respect to a triplet set $T$ if $X \neq L(T)$
and $L(T)$ is the only SN-set that is a strict superset of $X$. It is important in the following section to remember
that SN-sets are determined by triplets, not by networks.

\section{Constructing level-2 networks from \emph{dense} triplet sets is polynomial-time solvable}
\label{sec:poly}

We begin with some important lemmas and definitions.

\begin{lemma}
\label{lem:densepar} (Jansson and Sung, \cite{JS1}) If $T$ is dense, then for any $A, B \subseteq L(T)$, $SN_T(A) \cap
SN_T(B)$ equals $\emptyset$, $SN_T(A)$ or $SN_T(B)$.
\end{lemma}
From this lemma follows that the maximal SN-sets of $T$ partition $L(T)$. The next lemma shows that each SN-set is
equal to a set of the form $SN_T(\{x,y\})$ or $SN_T(\{x\})$, showing that we can find all SN-sets by the algorithm in
\cite{JS1}.

\begin{lemma}
\label{lem:2born} If $Y$ is an SN-set then $Y = SN_T(\{x,y\})$ or $Y = SN_T(\{x\})$ for some $x,y \in Y$.
\end{lemma}
\begin{proof}
Suppose $Y = SN_T(X)$. The proof is by induction on $|X|$. If $|X| \leq 2$ we are done. If $|X| > 2$ we take any three
leaves $x,y,z \in X$ such that $xy|z$ is a triplet in $T$. This is possible because $T$ is dense. We have thus that $Y
= SN_T(X) = SN_T(X \setminus \{x\})$ and the lemma follows by induction. $\Box$
\end{proof}

\noindent An arc $a=(u,v)$ of a network $N$ is a \emph{cut-arc} if its removal disconnects $N$. We write ``the set of
leaves below $a$'' to mean the set of leaves reachable from $v$ in $N$ and ``the set of vertices below $a$'' to mean
the set of all vertices reachable from $v$ in $N$.

\begin{lemma}
\label{lem:cutSN} Let $N$ be a network consistent with dense triplet set $T$. Then for each cut-arc $a$ in $N$, the
set $S$ of leaves below $a$ is an SN-set of $T$.
\end{lemma}
\begin{proof}
Clearly $SN_T(S) = S$, since the only triplet with leaves $x,y \in S$ and $z \notin S$ which is consistent with $N$ is
$xy|z$. $\Box$
\end{proof}

\noindent We say that a cut-arc $a = (u,v)$ is \emph{trivial} if $v$ is a leaf. We say that a cut-arc $a = (u,v)$ is a
\emph{highest} cut-arc if there does not exist a second cut-arc $a' = (u', v')$ such that $u$ is reachable from $v'$.

\begin{lemma}
\label{lem:highcut}
The sets of leaves below highest cut-arcs partition $L$.
\end{lemma}
\begin{proof}
Clearly every leaf is below a cut-arc, so it must also be below
some highest cut-arc. By the definition of highest cut-arc a leaf
cannot be below two highest cut-arcs. $\Box$
\end{proof}

\pagebreak

\begin{lemma}
\label{lem:integersum} Let $N$ be a network consistent with dense triplet set $T$. Each maximal SN-set $S$ in $T$ can
be expressed as the union of the leaves below one or more highest cut-arcs in $N$.
\end{lemma}
\begin{proof}
Suppose, by contradiction, that we have a maximal SN-set $S$ that does not have such a property. Clearly by
Lemma~\ref{lem:cutSN} $S$ cannot be a strict subset of the leaves below some single highest cut-arc $a$. Combining
this with Lemma \ref{lem:highcut} we conclude that $S$ intersects with the leaves below at least two highest cut-arcs.
It follows that there exist leaves $x, y, z$ such that $x$ is below highest cut-arc $a_1$, $y$ and $z$ are both below
highest cut-arc $a_2$, $x,z \in S$ and $y \notin S$. However, in this case the only triplet in $T$ on the leaves $x,
y, z$ is $yz|x$, meaning that $y \in S$ and thus yielding a contradiction. $\Box$
\end{proof}

\subsection{Simple level-2 networks}
\label{subsec:sim2}

We now introduce a class of level-2 networks that we name \emph{simple} level-2 networks. Informally these are the
basic building blocks of level-2 networks in the sense that each biconnected component of a level-2 network is in
essence a simple level-2 network. A simple characterisation of simple level-2 networks will be given in
Lemma~\ref{lem:withoutcut}. For the definition we first introduce a \emph{simple level-$k$ generator} (for $k\geq 1$),
which is defined as a directed acyclic multigraph:

\begin{enumerate}
\item that is biconnected; \item has a single root (indegree 0, outdegree 2), precisely $k$ recombination vertices
(indegree 2, outdegree at most 1) and apart from that only split vertices (indegree 1, outdegree 2),
\end{enumerate}

\noindent where vertices with indegree 2 and outdegree 0 as well as all arcs are labelled and called \emph{sides}.

\begin{definition}
A simple level-$k$ network $N$, for $k\geq 1$, is a network obtained by applying the following transformation (``leaf
hanging'') to some simple level-$k$ generator such that the resulting graph is a valid network:
\begin{enumerate}
\item replace each arc $X$ by a path and for each internal vertex $v$ of the path add a new leaf $x$ and an arc
$(v,x)$; we say that ``leaf $x$ is on side $X$''; and \item for each vertex $Y$ of indegree 2 and outdegree 0 add a
new leaf $y$ and an arc $(Y,y)$; we say that ``leaf $y$ is on side $Y$''.
\end{enumerate}
\end{definition}
Note that in the above transformation we obtain a valid network if and only if, whenever there are multiple arcs, we
replace at least one of them by a path of at least three vertices. A simple case-analysis (see
Lemma~\ref{lem:enumgens} in the appendix) shows that there is precisely one simple level-1 generator, and precisely
four simple level-2 generators, shown respectively in Figures~\ref{fig:lev1gen} and \ref{fig:lev2gen}. A simple
level-2 network built by hanging leaves from generator $8a$, $8b$, $8c$ or $8d$ is called a \emph{network
of type} $8a$, $8b$, $8c$ or $8d$ respectively.\\
\\
We do not attempt to define simple level-0 networks; instead we introduce the \emph{basic tree} which we define as the
directed graph on three vertices $\{v_1, v_2, v_3\}$ with arc set $\{ (v_1, v_2), (v_1,v_3) \}$. For the sake of
convenience we say that the basic tree, simple level-1 networks and simple level-2 networks are all \emph{simple
level-$\leq$2 networks.}


\begin{lemma}
\label{lem:withoutcut} A strict level-$k$ network is a simple level-$k$ network if and only if it contains no
nontrivial cut-arcs.
\end{lemma}
\begin{proof}
A simple level-$k$ network contains no nontrivial cut-arcs because simple level-$k$ generators are biconnected. Now
take a strict level-$k$ network $N$ with no nontrivial cut-arcs. Delete all leaves and suppress all vertices with
indegree and outdegree equal to one. The resulting graph is biconnected because any graph with degree at most three
containing a cut-vertex also contains a cut-arc. This graph is moreover a strict level-$k$ network and therefore it
has $k$ recombination vertices. It follows that this graph is a simple level-$k$ generator. $\Box$.
\end{proof}

\begin{figure}\centering
\vspace{-.2cm}\includegraphics{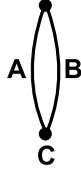} \caption{The only simple level-1 generator.}
\vspace{-.1cm}\label{fig:lev1gen}
\end{figure}
\begin{figure}\centering
\vspace{-.2cm}\includegraphics{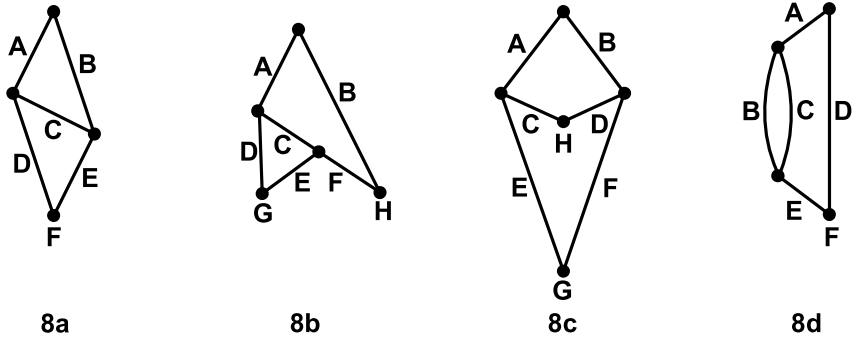} \caption{The four simple level-2
generators.}\vspace{-.1cm} \label{fig:lev2gen}
\end{figure}

\subsection{Constructing simple level-2 networks}

All simple level-1 networks can be found by the algorithm in \cite{JS2}. In this section we describe an algorithm that
constructs all simple level-2 networks consistent with a dense set of triplets. We start with some definitions. We say
that $z$ is a \emph{low leaf} of a network $N$ if its parent is a sink in $N\setminus L$. If arc $a$ enters a low leaf
$z$ we say that $a$ is a \emph{low arc} of $N$. An arc leaving the root or a child of the root is called a \emph{high
arc}. Recall that a \emph{recombination leaf} is a leaf whose parent is a recombination vertex. A network is a
\emph{caterpillar} if the deletion of all leaves gives a directed path. We call a set of leaves $L'$ a
\emph{caterpillarset} w.r.t. $T$ if we can write $L'=\{\ell_1,\ldots,\ell_k\}$ such that $\{\ell_1,\ldots,\ell_i\}$ is
an SN-set of $T$ for all $1\leq i\leq k$ or if $L'=\emptyset$. Lemma~\ref{lem:caterpillarsets}
explains why we call these sets caterpillarsets.\\
\\
Theorem~\ref{thm:SL2} shows that all simple level-2 networks can be found by Algorithm~\ref{alg:SL2}: SL2, which uses
the subroutine FindCaterpillarsets in Algorithm~\ref{alg:findcaterpillarsets}.

\begin{lemma}\label{lem:caterpillarsets}
Suppose that network $N$ consistent with dense triplet set $T$ contains (as a subgraph) a caterpillar with leaves
$L'$. Then is $L'$ a caterpillarset w.r.t. $T$.
\end{lemma}
\begin{proof}
Write $L'=\{\ell_1,\ldots,\ell_k\}$ such that $\ell_1$ has distance $k-1$ and $\ell_i$ ($2\leq i\leq k$) has distance
$k-i+1$ from the root of the caterpillar. Then by Lemma~\ref{lem:cutSN} the set $\{\ell_1,\ldots,\ell_i\}$ is an
SN-set of $T$ for all $1\leq i\leq k$. $\Box$
\end{proof}

\begin{algorithm}[h]
\caption{SL2} \label{alg:SL2}
\begin{algorithmic} [1]
\STATE $Net:=\emptyset$\\
\FOR{each leaf $x\in L$}
\STATE $L' := L\setminus\{x\}$\\
\STATE $T' := T | L'$\\
\STATE Construct all caterpillarsets w.r.t. $T'$ by the subroutine FindCaterpillarsets.\\
\FOR{each caterpillarset $C$}
\STATE $L'':=L'\setminus C$\\
\STATE $T'':=T'|L''$\\
\STATE Build the unique tree $N=(V,A)$ consistent with $T''$ if it exists by the algorithm in \cite{aho}.
\STATE $V:=V\cup\{r'\}$; $A:=A\cup\{(r',r)\}$ \COMMENT{with $r$ the root of $N$ and $r'$ a new dummy root}\\
\FOR{every arc $a_1=(u_1,v_1)$ and every low or high arc $a_2=(u_2,v_2)$ in $A$}\label{line:arcsy} \IF{$|C|\geq 2$}
\STATE Construct the caterpillar $(V_{cat},A_{cat})$ consistent with $T|C$ and let y be its root.\\
\ELSIF{$|C|=1$}
\STATE $V_{cat}:=C$, $A_{cat}:=\emptyset$ and let $y$ be such that $C=\{y\}$.\\
\ELSE
\STATE $V_{cat}:=\{y\}$, $A_{cat}:=\emptyset$ \COMMENT{with $y$ a dummy leaf}\\
\ENDIF
\STATE \COMMENT{Subdivide $a_1$ and $a_2$ and put the caterpillar below the new vertices as follows:}\\
\STATE $V':=V\cup V_{cat}\cup\{w_1, w_2, y'\}$\\
\STATE $A':=A\cup A_{cat}\setminus \{a_1,a_2\} \cup\{(u_i,w_i),(w_i,v_i),(w_i,y')|i=1,2\}\cup\{(y',y)\}$\\
\FOR{every two arcs $a_3=(u_3,v_4)$ and $a_4=(u_4,v_4)$ in $A'$}\label{line:arcs}
\STATE \COMMENT{Subdivide $a_3$ and $a_4$ and make $x$ a recombination leaf below the new vertices as follows:}\\
\STATE $V'':=V'\cup\{w_3, w_4, x', x\}$\\
\STATE $A'':=A'\setminus \{a_3,a_4\} \cup\{(u_i,w_i),(w_i,v_i),(w_i,x')|i=3,4\}\cup\{(x',x)\}$\\
\STATE $N'':=(V'',A'')$\\
\IF{$C=\emptyset$}
\STATE $N'':=N''\setminus\{y\}$ \COMMENT{remove the dummy leaf}
\STATE Suppress the former parent of $y$ with indegree and outdegree both equal to 1.\\
\ENDIF
\STATE $N'':=N''\setminus\{r'\}$ \COMMENT{remove the dummy root}\\
\IF{$N''$ is a simple level-2 network consistent with $T$}\label{line:consistency}
\STATE $Net:=Net\cup N''$\\
\ENDIF \ENDFOR \ENDFOR \ENDFOR \ENDFOR\\
\STATE \textbf{return} $Net$\\
\end{algorithmic}
\end{algorithm}
\begin{algorithm}[H]
\caption{FindCaterpillarsets} \label{alg:findcaterpillarsets}
\begin{algorithmic} [1]
\STATE $Cat := \{\emptyset\}$\\
\STATE Compute all SN-sets by the algorithm in \cite{JS1}.\\
\FOR{each singleton SN-set $S$}
\STATE $C := S$\\
\STATE $Cat := Cat\cup\{C\}$\\
\WHILE{there is an SN-set $C\cup \{x\}$ with $x\notin C$}
\STATE $C := C\cup\{x\}$\\
\STATE $Cat := Cat\cup\{C\}$\\
\ENDWHILE \ENDFOR
\STATE \textbf{return} $Cat$ \\
\end{algorithmic}
\end{algorithm}

\pagebreak

\begin{theorem}\label{thm:SL2}
Algorithm SL2 finds all simple level-2 networks consistent with a dense triplet set.
\end{theorem}
\begin{proof}
Consider any simple level-2 network. If we remove a recombination leaf $x$ and its parent $x'$, there is one
recombination vertex left and below it is a caterpillar, a single leaf or nothing (these are all caterpillarsets by
Lemma~\ref{lem:caterpillarsets} and will hence be found by Algorithm~\ref{alg:findcaterpillarsets}). If we
subsequently remove this recombination vertex and the caterpillar $Q$ below it, we obtain a tree which we can
construct using the algorithm of Aho et al. (and is unique, as shown in \cite{JS1}). From this tree we can reconstruct
the network if we choose the right arcs in the following procedure. We subdivide two specific arcs, connect the new
vertices to a new recombination vertex $y'$ which in turn is connected to the root of $Q$ and then subdivide two other
specific arcs, connect the new vertices to another recombination vertex $x'$ and connect $x'$ to $x$. Since on some
sides of the simple level-2 network there might be no leaves, we have to consider adding a dummy recombination leaf
(for when there are no leaves between the two recombination vertices) and add a
dummy root (for when there are no leaves on a side connected to the root).\\
\\
We will now prove that in line~11 we can always choose at least one high or low arc. First consider a network of type
$8a$. We can first remove the leaf on side $F$ and then the caterpillarset consisting of all leaves (possibly none) on
side $E$. If there are at least two leaves on side $B$ then one of the arcs we choose to subdivide is a low arc (we
subdivide the arc leading to the lowest leaf on side $B$). If on the other hand there is at most one leaf on side $B$
then one of the arcs we subdivide is a high arc (we subdivide the arc leaving the dummy root if there are no leaves on
side $B$ and we subdivide an arc leaving the child of the dummy root if there is exactly one leaf on
side $B$).\\
\\
Now consider a network of type $8b$. We first remove the leaf on side $G$ and the caterpillarset consisting of the
leaf on side $H$. Then we can argue just like with $8a$ that if there are at least two leaves on side $B$ we subdivide
a low arc and otherwise a high arc. In a network of type $8c$ we remove the leaf on side $G$ and the caterpillarset
consisting of the leaf on side $H$. If on one of the sides $C$, $D$, $E$ or $F$ there are at least two leaves we can
subdivide a low arc on this side. Otherwise there is at most one leaf on each of the sides $C$, $D$, $E$ and $F$ and
therefore all arcs we want to subdivide are low. Finally, consider a network of type $8d$. We remove the leaf on side
$F$ and the caterpillarset consisting of the leaves on side $E$. Then we can always subdivide a low arc unless there
are no leaves on sides $B$ and $C$, which is not allowed. This concludes the proof that in line~11 we
can always choose one high or low arc.\\
\\
Now suppose that triplet set $T$ is consistent with some simple level-2 network $N$. At some iteration the algorithm
will choose the right leaf and caterpillarset to remove and the right arcs to subdivide and the algorithm will
construct the network $N$. Furthermore, the algorithm checks in line~32 whether the output network is a simple level-2
network consistent with $T$. We conclude that the algorithm will find exactly all simple level-2 networks consistent
with $T$. $\Box$
\end{proof}

\begin{lemma}\label{lem:arcs}
Any level-2 network with $n$ leaves has $O(n)$ arcs.
\end{lemma}
\begin{proof}
The proof is deferred to the appendix. $\Box$
\end{proof}

\begin{lemma}\label{lem:consistency}
Given a level-2 network $N$ and a set of triplets $T$ one can decide in time $O(n^3)$ whether $N$ is a simple level-2
network consistent with $T$.
\end{lemma}
\begin{proof}
The proof is deferred to the appendix. $\Box$
\end{proof}

\begin{lemma}\label{lem:SL2runtime}
Algorithm SL2 can be implemented to run in time $O(n^8)$.
\end{lemma}
\begin{proof}
The number of caterpillarsets and arcs (by Lemma~\ref{lem:arcs}) are both $O(n)$ and the number of low and high arcs
is $O(1)$. Line~32 can be executed in time $O(n^3)$ by Lemma~\ref{lem:consistency}. Therefore, the algorithm can be
implemented to run in time $O(n^8)$. $\Box$
\end{proof}

\subsection{From simple to general level-2 networks}

This section explains why we can build level-2 networks by recursively building simple level-$\leq 2$ networks.\\
\\
Let $T$ be a dense set of triplets and $N$ a level-2 network consistent with $T$. Define $Collapse(N)$ as the network
obtained by, for each highest cut-arc $a = (u,v)$, replacing $v$ and everything reachable from it by a single new leaf
$V$, which we identify with the set of leaves below $a$. Let $L'$ be the leaf set of $Collapse(N)$. We define a new
set of triplets $T'$ on the leaf-set $L'$ as follows: $XY|Z \in T'$ if and only if there exists $x \in X, y\in Y, z
\in Z$ such that $xy|z \in T$. The fact that $T$ is dense implies that $T'$ is also dense. We write $T' =
CutInduce(N,T)$ as shorthand for the above. Observe that $CutInduce$ is a specific example of \emph{inducing} a new
triplet set using some given partition of the original triplet leaf set, in this case the partition created by highest
cut-arcs.\\

\vspace{1cm}

\pagebreak

\begin{lemma}
\label{lem:central}
(1) Any network $N'$ consistent with $T'$ can be transformed into a network
consistent with $T$ by ``expanding'' each leaf $V$ back into the subnetwork of $N$
that collapsed into it;\\
(2) $Collapse(N)$ is a simple level-$\leq$2 network and is consistent with $T'$.\\
(3) There is a bijection between the maximal SN-sets of $T$ and the maximal SN-sets of $T'$. Namely, a maximal SN-set
$S'$ in $T'$ corresponds to the maximal SN-set $S$ in $T$ obtained by replacing each leaf in $S'$ by the set of leaves
below the corresponding highest cut-arc in $N$.
\end{lemma}
\begin{proof}
(1) Suppose we have a network $N'$ consistent with $T'$ and a triplet $xy|z$ in $T$ such that the described expansion
of $N'$ is not compatible with $xy|z$. It cannot be true that $x, y, z$ all came from underneath the same highest
cut-arc in $N$, because the network structure underneath a highest cut-arc is left unchanged by the contractions and
expansions described above. It is not possible that $x, z$ originated from underneath one highest cut-arc and $y$ from
another because then $xy|z \notin T$. It is also not possible that $x,y$ came from underneath one highest cut-arc, and
$z$ from another, since in this case $xy|z$ is consistent with the expansion. So it must be that $x, y, z$ each
originated from the leaves below three different highest cut-arcs in $N$, say $X, Y, Z$ respectively. But $XY|Z$ would
then have been in $T'$, and $N'$ is consistent with $T'$, meaning that the expansion would indeed have been compatible
with $xy|z$, contradiction. (2) It is not difficult to see that $Collapse(N)$ is consistent with $T'$. That
$Collapse(N)$ is a simple level-$\leq$2 network, follows from Lemma~\ref{lem:withoutcut}. (3) Consider a maximal
SN-set $S$ of $T$. By Lemma~\ref{lem:integersum} $S$ is the union of the leaves below one or more highest cut-arcs in
$N$. Define the mapping $\delta$ as $\delta(S) = \{V\in L'~|~V\subseteq S\}$ and let $S'=\delta(S)$. We firstly show
that $S'$ is an SN-set for $T'$. Observe that $S' \neq L(T')$ because $S \neq L(T)$. Secondly, observe that $SN_T(S')
= S'$. If this would not be true then there would exist $X, Z \in S'$ and $Y \notin S'$ such that $XY|Z \in T'$. But
this would mean that there exist leaves $x, y, z \in L(T)$ and a triplet $xy|z \in T$ such that $x,z \in S$ and $y
\notin S$, yielding a contradiction. Now, consider a maximal SN-set $S'$ within $T'$, and let $S$ be the set of leaves
in $N$ obtained by expanding each leaf in $S'$ (i.e. the inverse mapping $\delta^{-1}$ of $\delta$) to obtain $S$.
Clearly $S$ is not equal to $L(N)$. It is also clear that $SN_T(S)=S$ because the existence of a triplet $xy|z$ such
that $x,z \in S$ and $y \notin S$ would mean the existence of a corresponding triplet $XY|Z$ in $T'$, violating the
maximality of $S'$. Now, suppose $S$ is a strict subset of some other SN-set in $T$; without loss of generality it is
thus a subset of a maximal SN-set $S^{*}$ in $T$. But by the above mapping $\delta$ we know that the existence of
$S^{*}$ guarantees the existence of an SN-set in $T'$ that is a superset of $S'$, giving a contradiction. So, there
remains only to show that the mapping $\delta$ maps maximal SN-sets of $T$ to maximal SN-sets of $T'$. Suppose this is
not true and some maximal SN-set $S$ in $T$ gets mapped to a non-maximal SN-set $S'$ in $T'$. Then $S'$ is a subset of
some maximal SN-set $S''$ in $T'$. The existence of $S''$ and the mapping $\delta^{-1}$ guarantees the existence of an
SN-set in $T$ which is a strict superset of $S$, giving a contradiction. $\Box$
\end{proof}

\begin{theorem}\label{thm:one}
Let $T$ be a dense triplet set consistent with some simple level-$\leq 2$ network $N$. Then there exists a level-2
network $N'$ consistent with $T$ such that at most one maximal SN-set $S$ of $T$ equals the union of the sets of
leaves below two cut-arcs and each other maximal SN-set is equal to the set of leaves below just one cut-arc.
\end{theorem}
\begin{proof}
The proof is rather complicated and is thus deferred to the appendix. $\Box$
\end{proof}

\noindent The following corollary proves that the above theorem also applies to general level-2 networks.

\begin{corollary}
\label{cor:theorem2} Let $T$ be a dense triplet set consistent with some level-2 network $N$. Then there exists a
level-2 network $N'$ consistent with $T$ such that at most one maximal SN-set $S$ of $T$ equals the union of the sets
of leaves below two cut-arcs and each other maximal SN-set is equal to the set of leaves below just one cut-arc.
\end{corollary}
\begin{proof}
Let $N^{*} = Collapse(N)$ and let $T' = CutInduce(N,T)$. By Lemma \ref{lem:central} we know that $N^{*}$ is simple
level-$\leq$2, is consistent with $T'$, and that there is a bijection between the maximal SN-sets of $T'$ and the
maximal SN-sets of $T$. By Theorem \ref{thm:one} there exists a network $N^{**}$ consistent with $T'$ with the desired
property. Replacing each leaf in $N^{**}$ by the subnetwork (of $N$) that collapsed into it gives a network $N'$
consistent with $T$ with the desired property. $\Box$
\end{proof}
The following theorem explains why, in essence, the entire algorithm can be reduced to the problem of finding simple
level-$\leq$2 networks. For a set of triplets $T$ and a set of SN-sets $M=\{S_1,\ldots,S_q\}$ we denote by $T\nabla M$
the set of triplets $S_i S_j |S_k$ such that there exist $x\in S_i$, $y\in S_j$, $z\in S_k$ with $xy|z\in T$ and $i$,
$j$ and $k$ all distinct.

\begin{theorem}
\label{thm:theorem3} Let $T$ be a dense set of triplets and $N'$ a network with the properties described in
Corollary~\ref{cor:theorem2}. Let $M$ be the set of SN-sets that are equal to the set of leaves below a highest
cut-arc. Then there exists a simple level-$\leq$2 network $N''$ consistent with $T\nabla M$. Furthermore, for any
simple level-$\leq$2 network $N''$ consistent with $T\nabla M$ holds that expanding each leaf into the subnetwork of
$N'$ that collapsed into it gives a level-2 network consistent with $T$.
\end{theorem}
\begin{proof}
Consider the network $N'' = Collapse(N')$ and the triplet set $T' = CutInduce(N', T)$. By Lemma \ref{lem:central} we
know that $N''$ is simple level-$\leq$2. The $CutInduce(.)$ function and the construction of $T\nabla M$ are in this
case identical, so $N''$ is clearly consistent with $T\nabla M$. Applying the first part of Lemma \ref{lem:central}
shows that $N''$ (and in fact any solution for $T'$) can be transformed back into a solution for $T$. $\Box$
\end{proof}
It remains to show how to find the set $M$ of SN-sets that are equal to the set of leaves below a highest cut-arc in
$N'$, when $N'$ is unknown. By Corollary~\ref{cor:theorem2} we know that, with the exception of possibly one maximal
SN-set of $T$, there is a one-to-one correlation between the elements of $M$ and the maximal SN-sets of $T$. Given
that there is at most one maximal SN-set that needs to be split into two pieces, we can simply try splitting each
maximal SN-set of $T$ in turn, as well as trying the case where no maximal SN-sets of $T$ are split. This does not
take too long because there is at most a linear (in the number of leaves in $T$) number of maximal SN-sets. The
following lemma tells us how to split the chosen maximal SN-set into two pieces.
\begin{lemma}\label{lem:splitSN}
Let $T$ be a dense set of triplets and $N'$ a network with the properties described in Corollary~\ref{cor:theorem2}.
Suppose $T$ contains a maximal SN-set $S$ which occurs as the union of the sets $S_1$ and $S_2$ of leaves below two
cut-arcs. Then $T|S$ contains precisely two maximal SN-sets and these are $S_1$ and $S_2$.
\end{lemma}
\begin{proof}
By Lemma~\ref{lem:cutSN}, $S_1$ and $S_2$ are both SN-sets of $T$. It is quite easy to see that $S_1$ and $S_2$ remain
SN-sets in the restriction of $T$ to $S$. Now, the fact that $S = S_1 \cup S_2$ means that $S_1$ and $S_2$ are both
maximal (in the restriction to $S$). To see why this is, consider that any alternative partition of $S$ into two or
more maximal SN-sets must contain at least one set that is a strict subset of either $S_1$ or $S_2$, contradicting the
maximality of the SN-sets. $\Box$
\end{proof}

\noindent Lemma~\ref{lem:cutSN} shows that only maximal SN-sets that internally decompose into two maximal SN-sets
need to be considered as possible candidates for the ``split'' SN-set, and furthermore that this split is totally
determined by the internal decomposition of the SN-set.

\subsection{Constructing a level-2 network}

\noindent We are finally ready to describe the complete algorithm. The general idea is as follows. We compute the
maximal SN-sets. If there are precisely two maximal SN-sets then we recursively create two level-2 networks for the
two maximal SN-sets and connect their roots to a new root. If there are three or more maximal SN-sets we try splitting
each maximal SN-set in turn, as well as that we try the case where no maximal SN-set is split. Lemma~\ref{lem:splitSN}
tells us how to split the maximal SN-set. If $M$ is the obtained set of SN-sets then we try to construct a simple
level-$\leq 2$ network $N$ consistent with $T\nabla M$. We recursively create level-2 networks for each SN-set in $M$
and replace each leaf of $N$ by the corresponding, recursively created, level-2 network. The complete pseudo code is
displayed in Algorithm~\ref{alg:L2}: L2.\\

\begin{algorithm}[H]
\caption{L2} \label{alg:L2}
\begin{algorithmic} [1]
\STATE $N:=null$\\

\STATE Compute the set of maximal SN-sets $SN$ of $T$ by the algorithm in \cite{JS1}.\label{line:SNsets1}\\

\IF{$|SN|=2$}

\STATE $N$ is a basic tree with leaves labelled $S_1$ and $S_2$.\\
\STATE $M^*:=SN$\\

\ELSE

\FOR{$S\in SN\cup\{\emptyset\}$} \label{line:beginforloop}

\STATE Compute the set of maximal SN-sets $SN'$ of $T|S$.\label{line:SNsets2}\\

\IF{$|SN'|=2$ \textbf{or} $S=\emptyset$}

\STATE $M := SN \setminus \{S\} \cup SN'$\\

\STATE $q:=|M|$\\

\STATE $T':=T\nabla M$\label{line:nabla}\\

\IF{$T'$ is consistent with a simple level-1 network} \label{line:beginif}

\STATE Let $N$ be such a network.\\
\STATE $M^*:=M$\\

\ELSIF{$T'$ is consistent with a simple level-2 network}

\STATE Let $N$ be such a network.\\
\STATE $M^*:=M$\\

\ENDIF \label{line:endif}

\ENDIF

\ENDFOR \label{line:endforloop}

\ENDIF

\STATE Denote the elements of $M^*$ by $S_1,\ldots,S_q$.\\

\FOR{$i=1,\ldots,q$}

\IF{$|S_i|>2$}

\STATE $N_i:=L2(T|S_i)$\\

\ELSE

\STATE $N_i$ is a basic tree with leaves labelled by the elements of $S_i$.\\

\ENDIF

\ENDFOR

\IF{any $N_i$ or $N$ equals $null$}

\STATE \textbf{return} $null$

\ENDIF

\FOR{$i=1,\ldots,q$}

\STATE Remove $S_i$ from $N$ and connect the former parent of $S_i$ to the root of $N_i$.\\

\ENDFOR

\STATE \textbf{return} $N$\\

\end{algorithmic}

\end{algorithm}

\begin{theorem}
Algorithm L2 constructs, in $O(|T|^3)$ time, a level-2 network consistent with a dense set of triplets if and only if
such a network exists.
\end{theorem}\begin{proof}
Correctness of the algorithm follows from Theorems~\ref{thm:SL2}~and~\ref{thm:theorem3} and
Corollary~\ref{cor:theorem2}. It remains to analyze the running time. A simple level-1 network can be found by the
algorithm in \cite{JS2} in time $O(n^3)$ and a simple level-2 network by algorithm~SL2 in time $O(n^8)$ by
Lemma~\ref{lem:SL2runtime}. Therefore, line~13~-~19 of Algorithm~L2 take $O(|SN|^8)$ time. Consider the constructed
network $N$ and denote the number of arcs going out of the nontrivial biconnected components by $s_1,\ldots,s_m$. For
the $i$-th nontrivial biconnected component of $N$ lines~13~-~19 are executed $s_i+1$ times, each iteration taking
$O(s_i^8)$ time. We have that $s_1+\ldots+s_m=O(n)$ since the total number of arcs is $O(n)$ by Lemma~\ref{lem:arcs}.
Because $s_1^9+\ldots+s_m^9\leq (s_1+\ldots+s_m)^9$ the total time needed for lines~13~-~19 is $O((s_1+\ldots+s_m)^9)$
and hence $O(n^9)$. The computation of maximal SN-sets in line~8 takes $O(n^5)$ time and is executed $O(n^2)$ times.
The computation of $T\nabla M$ takes $O(n^4)$ time and is also executed $O(n^2)$ times. All other computations can be
done in $O(n^5)$ time and are executed $O(n)$ times. We conclude that the total running time of Algorithm~L2 is
$O(n^9)$ which is equal to $O(|T|^3)$. $\Box$
\end{proof}

\vspace{.5cm}

\section{Constructing level-2 networks from triplet sets is NP-hard}
\label{sec:nphard}

In this section we prove that for a general, not necessary dense, set of triplets $T$ it is NP-hard to decide whether
there exists a level-2 network consistent with $T$. This is a nontrivial extension of the proof that this problem is
NP-hard for level-1 networks \cite{JS2}. Let $\tilde{N}$ be the network in Figure \ref{fig:unique} and $\tilde{T}$ the
set of all triplets consistent with $\tilde{N}$. For the NP-hardness proof we need to show that $\tilde{N}$ is the
only level-2 network consistent with $\tilde{T}$. We prove this in the Lemmas~\ref{lem:8csimple}~-~\ref{lem:8cunique}
in the appendix.

\begin{figure}\centering\vspace{-.5cm}
\includegraphics{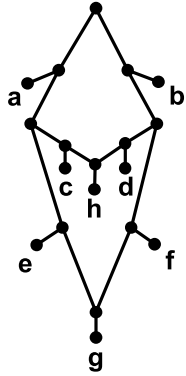}
\caption{The network $\tilde{N}$.}\vspace{-.5cm} \label{fig:unique}
\end{figure}

\begin{theorem}\label{thm:nphard}
It is NP-hard to decide whether for a given set of triplets $T$ there exists some level-2 network $N$ consistent with
$T$.
\end{theorem}
\begin{proof}
We reduce from the following NP-hard problem \cite{GareyJohnson}.\\
\\
\textit{Problem:} Set Splitting.\\
\textit{Instance:} Set $S=\{s_1,\ldots,s_n\}$ and collection $C=\{C_1,\ldots,C_m\}$ of cardinality-3 subsets
of $S$.\\
\textit{Question:} Can $S$ be partitioned into $S_1$ and $S_2$ (a set splitting) such that $C_j$ is not a subset of
$S_1$ or $S_2$ for all $1\leq j\leq m$?\\
\\
From an instance $(S,C)$ of Set Splitting we construct a set of triplets $T$ as follows. We start with the triplets
$\tilde{T}$ and for each set $C_j=\{s_a,s_b,s_c\}\in C$ (with $1 \leq a < b < c \leq n$) we add triplets
$s_a^jh|s_b^j$, $s_b^jh|s_c^j$ and $s_c^jh|s_a^j$. In addition, for every $s_i\in S$ and $1\leq j\leq m$ we add
triplets $hs_i^j|a$, $hs_i^j|b$, $he|s_i^j$, $hf|s_i^j$ and (if $j\neq m$) $s_i^js_i^{j+1}|h$. This completes the
construction of $T$ and we will now prove that $T$ is consistent with some level-2 network if and only if there exists
a set splitting $\{S_1,S_2\}$ of $(S,C)$.\\
\\
First suppose that there exists a set splitting $\{S_1,S_2\}$. Then we construct the network $N$ by starting with the
network $\tilde{N}$ and adding leaves to it as follows. For each element $s_i\in S_1$ we put all leaves $s_i^j$ (for
all $1\leq j\leq m$) between $a$ and the split vertex below $a$; for each element $s_i\in S_2$ all leaves $s_i^j$ (for
all $1\leq j\leq m$) are added between $b$ and the split vertex below $b$. To determine the order in which to put
these leaves consider a set $C_j=\{s_a,s_b,s_c\}\in C$. If $s_a$ and $s_b$ are on the same side of the partition we
put leaf $s_a^j$ below $s_b^j$, if $s_b$ and $s_c$ are on the same side of the partition we put $s_b^j$ below $s_c^j$
and if $s_a$ and $s_c$ are on the same side we put $s_c^j$ below $s_a^j$. The rest of the ordering is arbitrary. It is
easy to check that all triplets are indeed satisfied. For an example of this construction see
Figure~\ref{fig:hardreduc}.\\

\begin{figure}\centering
\vspace{-.5cm}\includegraphics{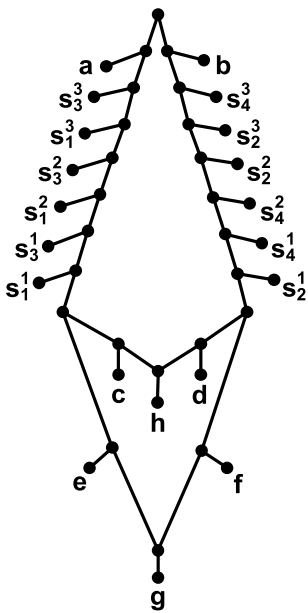} \caption{Example of the construction of the network $N$ in the proof of
Theorem~\ref{thm:nphard} for $C_1=\{s_1,s_3,s_4\}$, $C_2=\{s_2,s_3,s_4\}$ and $C_3=\{s_1,s_2,s_4\}$.}
\label{fig:hardreduc}\vspace{-.5cm}
\end{figure}

\noindent Now suppose that $T$ is consistent with some level-2 network $N$. Since $\tilde{T}\subset T$ we know by
Lemma~\ref{lem:8cunique} that $N$ must be equal to $\tilde{N}$ with the leaves not in $L(\tilde{N})$ added. From the
triplets $hs_i^j|a$ and $hs_i^j|b$ it follows that none of the leaves $s_i^j$ can be between $a$ and the root or
between $b$ and the root. In addition, from the triplets $he|s_i^j$ and $hf|s_i^j$ we know that $s_i^j$ cannot be
below any of the two split vertices. If follows that each $s_i^j$ must either be between $a$ and the left split vertex
or between the right split vertex and $b$. In addition, from the triplets $s_i^js_i^{j+1}|h$ we know that for each
$1\leq i\leq n$ all $s_i^j$ ($1\leq j\leq m$) have to be on the same side. Let $S_1$ be the set of elements $s_i\in S$
for which all $s_i^j$ ($1\leq j\leq m$) are between $a$ and the left split vertex and denote by $S_2$ the set of
elements $s_i\in S$ for which all $s_i^j$ ($1\leq j\leq m$) are between $b$ and the right split vertex. It remains to
prove that $(S_1,S_2)$ is a set splitting of $(S,C)$. Consider a set $C_j=\{s_a,s_b,s_c\}$ and suppose that
$s_a,s_b,s_c\in S_1$ (the case $s_a,s_b,s_c\in S_2$ is symmetric). It follows that all leaves $s_a^j,s_b^j,s_c^j$ are
somewhere between $a$ and the root. But since $T$ contains all triplets $s_a^jh|s_b^j$, $s_b^jh|s_c^j$ and
$s_c^jh|s_a^j$ this is not possible. $\Box$ \end{proof}

\section{Conclusion and open questions}
\label{sec:conclusion}

In this paper we have shown that it is polynomial-time solvable to construct level-2 networks when the input triplet
set is dense. In this way we have brought more complex, interwoven forms of evolution within reach of triplet methods.
There remain, of course, many open questions and challenges, which we briefly list here.

\begin{enumerate}
\item \textbf{Applicability.} A practical challenge is to implement the algorithm and to test it on real biological
data. How plausible are the networks that the algorithm constructs? How does it compare to the networks produced by
other packages? How far is the critique from certain parts of the community on the validity of many quartet-based
methods also relevant here? This critique in essence rests on the argument that it is in practice far harder to
generate high-quality input quartets than is often claimed. The \emph{short quartet method} \cite{short1} has been
discussed as a way of addressing this critique. This debate needs to be addressed in the context of this paper. \item
\textbf{Implementation.} Related to the above, it will be interesting to see how far the running time of our algorithm
can be improved and/or how far this is necessary for practical applications. At the moment it runs in time $O(|T|^3)$.
\item \textbf{Complexity I.} Is the dense level-$k$ problem always polynomial-time solvable for fixed $k$? As
discussed in the introduction it might in this regard be helpful to try generalising Theorem~\ref{thm:one}, which
captures the behaviour of SN-sets, and Theorem~\ref{thm:SL2}, which proves the polynomial-time solvability of
constructing simple level-2 networks. Generalising Theorem~\ref{thm:one} will probably be difficult, because it is at
this moment not clear whether the technique of ``pushing'' maximal SN-sets below cut-arcs generalises to level-3 and
higher. \item \textbf{Complexity II.} What is the computational complexity of the following problem: Given a dense set
of triplets $T$, compute the smallest $k$ for which there exists a level-$k$ network $N$ that is consistent with $T$.
\item \textbf{Complexity III.} Confirm the conjecture that non-dense level-$k$ is NP-hard for all fixed $k \geq 1$.
\item \textbf{Bounds.} In \cite{JS2} the authors determine constructive lower and upper bounds on the value $p$ for
which the following statement is true: for each set of triplets $T$, not necessarily dense, there exists some level-1
network $N$ which is consistent with at least $p|T|$ triplets in $T$. It will be interesting to explore this question
for level-2 networks and higher. \item \textbf{Building all networks.} It is not clear whether it is possible to adapt
our algorithm to generate \emph{all} level-2 networks consistent with the input triplet set. If so, then such an
adaptation could (even in the case that exponentially many networks are produced) be very useful for comparing the
plausibility and/or relative similarity of the various solutions. \item \textbf{Properties of constructed networks.}
Under what conditions on the triplet set $T$ is there only one network $N$ for which $N$ is consistent with $T$? Under
what conditions does $T$ permit some solution $N$ such that the set of all triplets consistent with $N$, is exactly
equal to $T$? These questions are also valid for level-1 networks. \item \textbf{Different triplet restrictions.}
Density is only one of very many possible restrictions on the input triplets. A particularly interesting alternative
is what we have named \emph{extreme density}, which is strongly related to the previous point. Here we assume that the
input triplets were derived from some real network, and that \emph{all} triplets within that network were found, not
just some dense subset of them; this might be a plausible assumption if the applied triplet generation method is fast
and generates high-quality triplets. What is the complexity of reconstructing the original network or, indeed, any
network consistent with \emph{exactly} the set of input triplets? There are some indications that, because the input
is guaranteed to contain a large amount of information, such extreme density reconstruction problems might be easier
to reason about for higher-level networks. \item \textbf{Confidence.} At the moment all input triplets are assumed to
be correct. Is there scope for attaching a confidence measure to each input triplet, and optimising on this basis?
This is also related to the problem of ensuring that certain triplets are \emph{excluded} from the output network, as
explored in \cite{forbid}. \item \textbf{Exponential-time exact algorithms.} As shown in \cite{JS2} and in this paper
the general level-$k$ problem for $k \in \{1,2\}$ is NP-hard. It could be interesting, and useful, to develop
exponential-time exact algorithms for solving these problems.
\end{enumerate}

\section*{Acknowledgements}
We thank Katharina Huber for her useful ideas and many interesting discussions.

\vspace{3cm}

\pagebreak

\appendix

\section{Appendix}

\begin{lemma}
\label{lem:enumgens}
There is only 1 simple level-1 generator, and there are only 4 simple level-2 generators,
and these are shown in Figures \ref{fig:lev1gen} and \ref{fig:lev2gen} respectively.
\end{lemma}
\begin{proof}
To see that Figure \ref{fig:lev1gen} is the only simple level-1 generator, note firstly that a generator cannot
contain leaves. Hence, for each vertex in the generator, all paths beginning at that vertex must terminate at a
recombination vertex. A recombination vertex can end at most two paths, and a split vertex increases the number of
paths that still need to be ended by one. The root vertex introduces two paths and there is precisely one
recombination vertex, so the simple level-1 generator cannot contain any split vertices. The uniqueness of the
simple level-1 generator follows.\\
\\
There remains the case of the simple level-2 generators. By the above reasoning a simple level-2 generator can have at
most two split vertices; three or more split vertices would mean (if the two paths beginning at the root were
included) at least five paths would have to be ended, and two recombination vertices can only end at most four
paths. Similarly, a level-2 generator must have at least one split vertex.\\
\\
\textbf{Case 1: one split vertex.} Consider the two arcs leaving the root. It is not possible that they both end at
the same recombination vertex $r$ because then the removal of $r$ will disconnect the graph. So precisely one of these
arcs ends at a split vertex $s$ and the other at a recombination vertex. There are no other split vertices so both
arcs leaving $s$ must enter recombination vertices. The two possibilities
for this lead to $8a$ and $8d$ from Figure \ref{fig:lev2gen}.\\
\\
\textbf{Case 2: two split vertices.} Let $(r,x)$ and $(r,y)$ be the two arcs leaving the root. It is again not
possible that $x$ and $y$ are both equal to the same recombination vertex. Consider the case where $x$ and $y$ are
both equal to split vertices. This creates four paths that need to be ended, so all the arcs leaving $x$ and $y$ have
to enter recombination vertices. There is only one way to do this such that the graph is biconnected; this creates
$8c$. There remains only the case when (without loss of generality) $x$ is a split vertex and $y$ is a recombination
vertex. Consider the two arcs $(x,p)$ and $(x,q)$. It cannot be that $p$ and $q$ are both equal to the same
recombination vertex, because the existence of a second split vertex reachable from $p$ creates two paths that need to
be ended, and $y$ can only end one path. So (without loss of generality) $p$ is also a split vertex. Neither $p$ nor
$q$ can be equal to $y$ because then the resulting graph will not be biconnected. So one of the children of $p$ is
equal to $y$ and the other child of $p$ joins with the remaining child of $x$ to form the second recombination vertex.
This gives $8b$. $\Box$
\end{proof}

\noindent\textbf{Theorem \ref{thm:one}} \emph{Let $T$ be a dense triplet set consistent with some simple level-$\leq
2$ network $N$. Then there exists a level-2 network $N'$ consistent with $T$ such that for at most one maximal SN-set
$S$ of $T$ there exist two cut-arcs $a_1$ and $a_2$ in $N'$ such that $S$ equals the union of the sets of leaves below
$a_1$ and $a_2$ and each other maximal SN-set is equal to the set of leaves below just one cut-arc.}
\begin{proof}
Let $T$ a dense triplet set consistent with a simple level-$\leq 2$ network $N$ and $S$ a maximal SN-set of $T$. We
start with two critical observations.
\begin{observation}\label{obs:root}
No two leaves in $S$ can have only one common ancestor (the root). Since this would imply that all leaves are in $S$.
$\Box$
\end{observation}

\noindent We say that $u$ is a \emph{lowest common ancestor} of $x$ and $y$ if $u$ is an ancestor of both $x$ and $y$
and no proper descendant of $u$ has this property.

\begin{observation}\label{obs:LCA}
If two leaves $x,y\in S$ have exactly one lowest common ancestor $u$ then all leaves $z$ that have a parent on a path
from $u$ to $x$ have to be in $S$, unless $N$ is of type $8a$ with leaves on sides $C$ and $F$ in $S$ and no leaves on
side $E$ or $B$ in $S$: the situation from Figure~\ref{fig:thm1fig1}. $\Box$
\end{observation}

\begin{figure}\centering
\vspace{-.5cm}\includegraphics{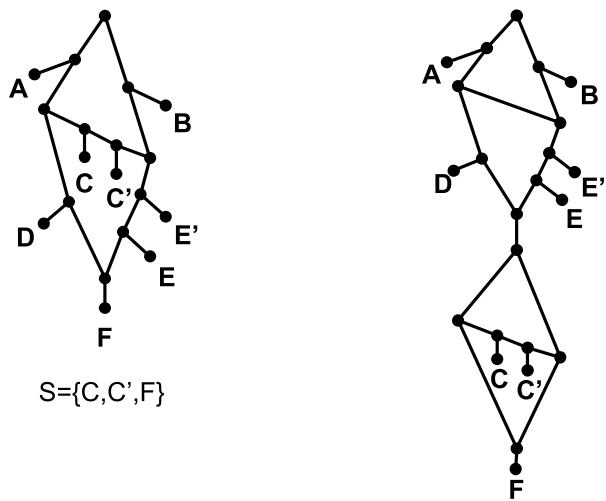} \caption{$8a$ with leaves on sides $C$ and $F$ in $S$ and no leaves on
side $E$ or $B$ in $S$}\vspace{-.3cm}\label{fig:thm1fig1}
\end{figure}

\noindent The case that $N$ is a basic tree is trivial. Now suppose that $N$ is a simple level-1 network. From the
above observations it follows that $S$ equals the set of leaves that have a parent on a path ending in either the
recombination vertex $q$ or in a parent $p$ of the recombination vertex. We can construct the network $N'$ by putting
the leaves in $S$ on a caterpillar below $q$ or $p$ respectively. From now on we assume that $N$ is a simple level-2
network and we prove the following lemma.
\begin{lemma}\label{lem:2paths}
There are at most three paths in $N$ such that each leaf in $S$ has a parent on one of these paths. In fact there are
only at most two such paths unless $N$ is of type $8b$ and leaves on sides $D$, $E$ and $F$ are in $S$.
\end{lemma}
\begin{proof}
First consider $8a$ and suppose that a leaf on side $B$ is in $S$. Then by Observation~\ref{obs:root} no leaves on
sides $A$, $C$ and $D$ can be in $S$ and the lemma follows. If no leaf on side $B$ is in $S$ then the lemma is also
clearly
true.\\
\\
For $8b$ we can argue similarly that if a leaf of side $B$ is in $S$ no leaves on sides $C$, $D$, $E$, $F$ and $G$ can
be in $S$. Hence if $B$ is or is not in $S$ the lemma is clearly true.\\
\\
Now consider $8c$ and suppose that a leaf on one of the sides $A$, $C$ or $E$ is in $S$. Then no leaves on sides $B$,
$D$ and $F$ can be in $S$ by Observation~\ref{obs:root}. In $8d$ we argue that if a leaf on side $D$ is in $S$ there
can only be leaves from sides $D$ and $F$ in $S$. In all cases their are at most two paths such that each leaf in $S$
has a parent on one of these paths. $\Box$
\end{proof}

\noindent If $N$ is indeed of type $8b$ and leaves on sides $D$, $E$ and $F$ are in $S$ then it follows by
Observation~\ref{obs:LCA} that in this case all leaves on sides $D$, $E$ and $F$ and on sides $C$ and $G$ are in $S$.
We call this Situation X and come back to this later.\\
\\
Now assume that not all leaves with a parent on one of these paths are in $S$. Then there are leaves $x,z\in S$ and
$y\notin S$ such that there is a path from the parent of $x$ to the parent of $y$ to the parent of $z$. But since
$xz|y\in T$ there must also be vertices $u$ and $v$ such that there are disjoint paths from $u$ to $x$, from $u$ to
$z$, from $u$ to $v$ and from $v$ to $y$. This is only possible in $8a$ when leaves on sides $C$ and $F$ are in $S$
and no leaves on side $E$ or $B$ are in $S$. In this case we can construct a network $N'$ with the desired properties
like in Figure~\ref{fig:thm1fig1}. Otherwise we may assume that there are at most two paths in $N$ such that $S$
consists
of exactly those leaves that have a parent on one of these paths.\\
\\
In simple level-2 structures there is only one possible situation where (the roots of) the two paths can have two
different lowest common ancestors. This is in $8c$ with $S=\{G,H\}$, in which case we're done. Otherwise there is a
unique lowest common ancestor (LCA) of these two paths and the union of this LCA and the two paths is, by
Observation~\ref{obs:LCA} a connected subgraph. In this case we can as well add this LCA to one of the paths. From now
on we assume that the union of the two paths is a connected subgraph $CSG[N,S]$. For the construction of $N'$ we need
the following important property of $CSG[N,S]$.
\begin{lemma}\label{lem:degrees}
The connected subgraph $CSG[N,S]$ has at most two outgoing and three incoming arcs. In addition, if it has two
outgoing arcs it has only one incoming arc.
\end{lemma}
\begin{proof}
Since there are only two recombination vertices it is certainly not possible to have more that three incoming arcs. We
will now determine the maximum number of outgoing arcs by considering the different simple level-2 structures
separately.\\
\\
Note that $8a$ and $8d$ have (excluding the root) only one split vertex. Hence any connected subgraph can have at most
three outgoing arcs. From Observation~\ref{obs:root} follows that $CSG[N,S]$ can have at most two outgoing arcs in
these structures.\\
\\
Simple level-2 networks of type $8c$ have two split vertices and a root. However, by Observation~\ref{obs:root} it is
not possible that $CSG[N,S]$ contains vertices on sides $A$, $C$ or $E$ and simultaneously on sides $B$, $D$ or $F$.
Hence also in networks of this type $CSG[N,S]$ can have at most two outgoing arcs.\\
\\
Finally consider $8b$. Here $CSG[N,S]$ can contain vertices on the sides $D$, $E$ and $F$ but in that case it contains
these sides completely and is also $G$ in $S$. It follows that also in this case $CSG[N,S]$ has at most two outgoing
arcs.\\
\\
Now suppose that $CSG[N,S]$ has two incoming arcs. That means that it contains a recombination vertex. In $8b$ and
$8c$ this has to be the parent of a leaf in $S$. In $8a$ and $8d$ it follows that whole of the side $E$ is contained
in $CSG[N,S]$. In each case there can clearly be at most one outgoing arc because of the limited number of split
vertices. $\Box$
\end{proof}
Now we can use the following procedure to construct $N'$. We remove all leaves from $S$ and contract $CSG[N,S]$ to a
single vertex $v_c$. Then we make a copy of the original network $N$ and remove all leaves that are not in $S$ from
this copy and we denote the result as $C[N,S]$. How we connect $C[N,S]$ to the network depends on $v_c$. From
Lemma~\ref{lem:degrees} we know that $v_c$ has outdegree at most two, indegree at most three and cannot have indegree
and outdegree both equal to two. If $v_c$ is a split vertex (i.e. indegree one and outdegree two) we subdivide the arc
entering $v_c$ and connect $C[N,S]$ by adding an arc from the new vertex to the root of $C[N,S]$. If $v_c$ has
indegree at least two and hence outdegree one we subdivide the arc leaving $v_c$ and connect $C[N,S]$ to this new
vertex. If $v_c$ has indegree three we replace it by two vertices of indegree two and if $v_c$ has outdegree zero or
indegree and outdegree one we make a new outgoing arc from $v_c$ to the root of $C[N,S]$. Finally we can simplify the
obtained network by removing unlabelled leaves (recombination vertices) and suppressing vertices
of degree two.\\
\\
This whole procedure is illustrated in Figure~\ref{fig:thm1fig2}. An example network $N$ is displayed on the left with
the two paths in red. After removing the leaves from $S=\{H,D,G,F\}$ and contracting the two paths to a single vertex
$v_c$ (in red) we get the network in the middle. Since $v_c$ has indegree three and outdegree zero we replace it by
two vertices of indegree two and create an outgoing arc to which we connect a copy of the original network, but
without the leaves from $S$. After suppressing all degree-2 vertices we get the network $N'$ on the right.\\

\begin{figure}\vspace{-.5cm}
\begin{minipage}{.3\textwidth}
\centering
\includegraphics{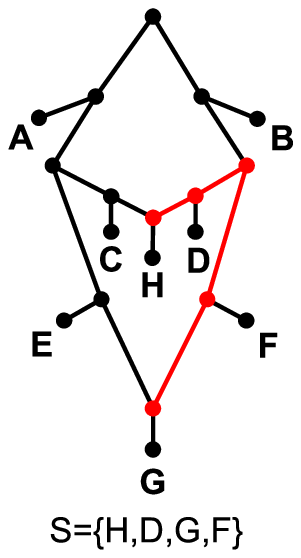}
\end{minipage}
\begin{minipage}{.3\textwidth}
\centering
\includegraphics{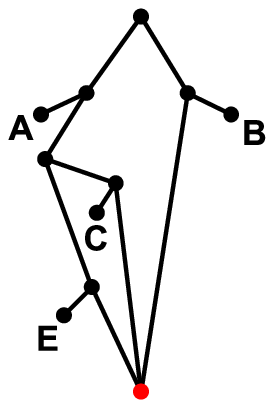}
\end{minipage}
\begin{minipage}{.3\textwidth}
\centering
\includegraphics{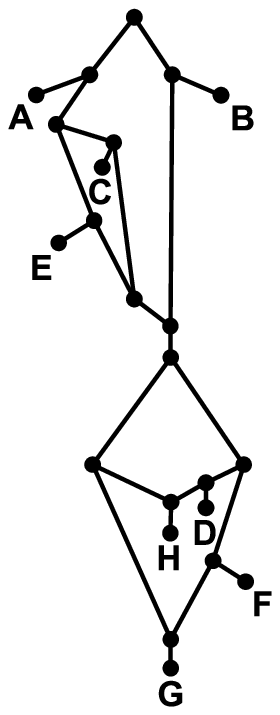}
\end{minipage}
\caption{Construction of $N'$ from $N$ in the proof of Theorem~\ref{thm:one}.}\vspace{-.3cm}\label{fig:thm1fig2}
\end{figure}

\noindent Also in Situation X we can use the same procedure. The result is illustrated in Figure~\ref{fig:thm1fig3},
for both the case that $H\in S$ (above) and that $H\notin S$ (below).\\

\begin{figure}
\centering\vspace{-.5cm}
\includegraphics{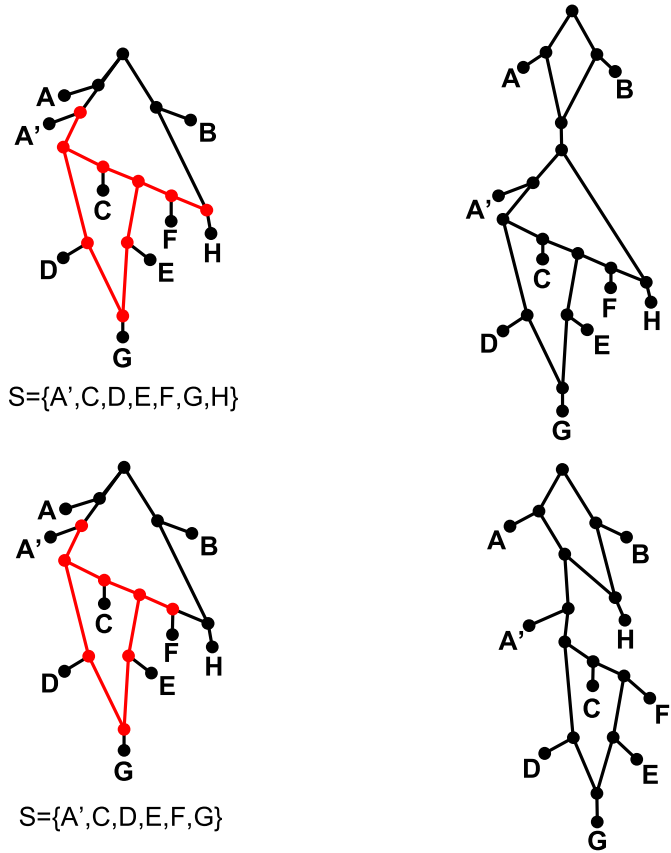} \caption{Construction of $N'$ in situation X for $H \in S$ (above) and $H \not
\in S$ (below).}\vspace{-.3cm}\label{fig:thm1fig3}
\end{figure}

\noindent To see that $T$ is consistent with each of the networks $N'$ above, consider three leaves $x,y,z$. If
$x,y,z\in S$ or $x,y,z\notin S$ then any triplet on these leaves that is consistent with $N$ is clearly also
consistent with $N'$. If $x,y\in S$ and $z\notin S$ then $xy|z$ is the only triplet in $T$ on these three leaves and
this triplet is consistent with $N'$. Finally, consider the case that $x\in S$ and $y,z\notin S$ and suppose that a
triplet $t$ over $\{x,y,z\}$ is consistent with $N$ but not with $N'$. Note that we only contracted and de-contracted
arcs. A de-contraction can never harm a triplet. And a contraction can only harm a triplet if it contracts the whole
path between the two internal vertices of the triplet. This means that just after the contractions (before the other
modifications) there are three disjoint paths from $v_c$ to $x$, $y$ and $z$. This means that $v_c$ is a split vertex
(after removing the leaves from $S$) and that the arc leading to $v_c$ is subdivided and $C[N,S]$ is connected to the
new vertex. It follows that the triplet $yz|x$ is consistent with $N'$. Now suppose that $t=xy|z$. We may assume that
$S$ contains another leaf $x'$ that is (in $N$) on the other side or above the split vertex splitting the two paths.
Because $xx'|z\in T$ it follows that there is a path from the root to $z$ not passing through $v_c$. If this path does
not intersect the path from $v_c$ to $y$ then this implies that $xy|z$ is consistent with $N'$. Otherwise, $N$ is of
type $8a$ with $x$ on side $C$, $x'$ on side $D$, $y$ on side $E$ and $z$ on side $F$. In this case $S$ contains by
Observation~\ref{obs:LCA} all leaves on sides $C$ and $D$ and maybe some on side $A$. In this case we can construct
the network $N'$ like in Figure~\ref{fig:thm1fig4}. The case $xz|y$ is symmetric.\\

\begin{figure}
\centering\vspace{-.5cm}
\includegraphics{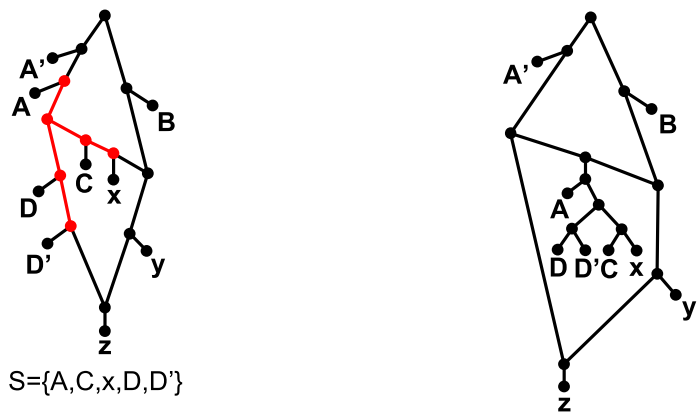} \caption{Construction of $N'$ in $8a$ with $x$ on side $C$, $x'$ on side $D$,
$y$ on side $E$ and $z$ on side $F$.}\vspace{-.3cm}\label{fig:thm1fig4}
\end{figure}

\noindent Repeating this procedure for each maximal SN-set gives a level-2 network $N'$ such that each maximal SN-set
equals the union of the sets of leaves below two cut-arcs. Furthermore, unless $N$ is of type $8c$, every maximal
SN-set equals the set of leaves below just one cut-arc in $N'$. If $N$ is of type $8c$ then this is also true except
for the (potential) maximal SN-set $\{G,H\}$. $\Box$
\end{proof}

\noindent \textbf{Lemma \ref{lem:arcs}} \emph{Any level-2 network with $n$ leaves has $O(n)$ arcs.}
\begin{proof}
The proof is by induction on $n$. Assume that any level-2 network with $n\leq M$ leaves has at most $8n$ arcs.
Consider a level-2 network $N$ with $M+1$ leaves. Observe that each nontrivial biconnected component of $N$ is a
simple level-1 or simple level-2 network after removing all leaves. Recall from Section \ref{sec:def} that for any
valid network always holds that every nontrivial biconnected component has at least three outgoing arcs. From this
follows that there always exists at least one leaf that is not a recombination leaf. Take any such a leaf, remove it
and suppress its parent with indegree and outdegree equal to one. If this creates a nontrivial biconnected component
with only two outgoing arcs, we replace it by a single split vertex. Otherwise, if we created multiple arcs we replace
these by a single arc and suppress the two obtained vertices with indegree and outdegree both equal to one (this only
occurs if we remove all leaves on sides $B$ and $C$ in a component of type $8d$). In each case, we reduce the number
of leaves by one and the number of arcs by at most eight. The obtained network has, by the induction hypothesis, at
most $8M$ arcs. Hence $N$ had at most $8M+8=8(M+1)$ arcs. $\Box$
\end{proof}

\noindent\textbf{Lemma \ref{lem:consistency}} \emph{Given a level-2 network $N$ and a set of triplets $T$ one can
decide in time $O(n^3)$ whether $N$ is a simple level-2 network consistent with $T$.}
\begin{proof}
Consider the following function $f$. For sides $X,Y,Z\in\{A,B,C,D,E, F, G, H\}$ and type $t\in\{8a,8b,8c,8d\}$ the
function $f(t,X,Y,Z)=1$ if a triplet $xy|z$ is consistent with a network of type $t$ with a leaf $x$ on side $X$, $y$
on side $Y$ and $z$ on side $Z$ such that if some of these leaves are on the same side $x$ and $y$ are always below
$z$ and $x$ is always below $y$. Otherwise, $f(t,X,Y,Z)=0$. If for a triplet $xy|z$ leaf $z$ is on the same side but
below $x$ or $y$ then this triplet is not consistent with the network. For any other triplet the function $f$ can be
used to evaluate the consistency of the triplet with any simple level-2 network. Furthermore, the function $f$ can be
computed in constant time. It remains to prove that one can determine, in $O(n^2)$ time, whether $N$ is a simple
level-2 network and if so to find the type of the simple level-2 network as well as the order of the leaves on the
sides. Subsequently one can use the function $f$ to decide for all $O(n^3)$ triplets whether they are consistent with
$N$.\\
\\
For determining whether $N$ is of type $8a$, $8b$, $8c$ or $8d$ we make use of the following subroutine
\emph{LeafPathBetween}$(p,q, A)$ where $p$ and $q$ are non-leaf vertices of $N$. The subroutine returns TRUE if there
is a path from $p$ to $q$ such that the only vertices ``hanging off'' that path (if any) are leaves, and such that no
internal vertices of the path are in $A$. In the case that the subroutine returns TRUE it also returns the internal
vertices of such a path. Otherwise it returns FALSE. It is easy to see that this subroutine executes in polynomial
time. Namely, we start a path at $p$ via an out-arc of $p$ (if there are two out-arcs from $p$ we simply try the
following algorithm for both out-arcs), and examine the current vertex $v$ on the path. If $v$ is equal to $q$ we are
done, return TRUE and the internal vertices of the path that we have constructed. If $v$ is in $A$ return FALSE.
Otherwise, consider the arcs entering and leaving $v$ that are \emph{not} equal to the arc that we entered $v$ by.
There are several mutually-exclusive cases to consider. (i) If there is an in-arc then return FALSE. (ii) If all
children are leaves return FALSE. (iii) If there are two children and neither are leaves, return FALSE. (iv) If there
are two children, one is a leaf and one is not a leaf, continue to the non-leaf child and iterate.\\
\\
Now, let us consider the question of determining whether $N$ has the form of $8a$
(respectively $8b$, $8c$, $8d$.)\\
\\
\textbf{Case $8a$: } In $8a$ there must be exactly one root vertex (indegree 0, outdegree 2), one \emph{internal}
split vertex (indegree 1, outdegree 2 such that neither child is a leaf), one \emph{internal} recombination vertex
(indegree 2, outdegree 1 such that the child is not a leaf) and one \emph{external} recombination vertex (indegree 2,
outdegree 1 such that the child \emph{is} a leaf.) It is clearly easy to check in polynomial time that this is so, so
let these four vertices be $v_1, v_2, v_3, v_4$
respectively. We then execute the following code.\\
\\
Set $SEEN=\emptyset$;\\
$P := LeafPathBetween(v_1, v_2, SEEN);$\\
If $P =$ FALSE return FALSE;\\
$SEEN := SEEN \cup P;$\\
$P := LeafPathBetween(v_1, v_3, SEEN);$\\
$SEEN := SEEN \cup P;$\\
If $P =$ FALSE return FALSE;\\
$P := LeafPathBetween(v_2, v_3, SEEN);$\\
$SEEN := SEEN \cup P;$\\
If $P =$ FALSE return FALSE;\\
$P := LeafPathBetween(v_2, v_4, SEEN);$\\
$SEEN := SEEN \cup P;$\\
If $P =$ FALSE return FALSE;\\
$P := LeafPathBetween(v_3, v_4, SEEN);$\\
If $P =$ FALSE return FALSE;\\
Return TRUE.\\
\\
\textbf{Case $8b$.} Here there is exactly one root vertex ($v_1$), two internal split vertices ($v_2$ and $v_3$), and
two external recombination vertices ($v_4$ and $v_5$). (Note that we can use \emph{LeafPathBetween}$(v_2, v_3,
\emptyset)$ to check whether $v_2$ is the ancestor of $v_3$ or vice-versa; it will return TRUE if $v_2$ is the
ancestor of $v_3$ and FALSE if not. Having identified $v_2$ and $v_3$ it is easy to again use \emph{LeafPathBetween}
to identify $v_4$ and $v_5$.) We can then use the same pseudocode as in case $8a$, this time making calls to
\emph{LeafPathBetween} in the following order: $v_1 \rightarrow v_2$, $v_1 \rightarrow v_5$, $v_2 \rightarrow
v_3$, $v_2 \rightarrow v_4$, $v_3 \rightarrow v_4$, $v_3 \rightarrow v_5$.\\
\\
\textbf{Case $8c$.} Here there is exactly one root vertex ($v_1$), two internal split vertices ($v_2$ and $v_3$), and
two external recombination vertices ($v_4$ and $v_5$). This time we make calls to \emph{LeafPathBetween} in the
following order: $v_1 \rightarrow v_2$, $v_1 \rightarrow v_3$, $v_2 \rightarrow v_4$, $v_3 \rightarrow v_4$,
$v_2 \rightarrow v_5$ and $v_3 \rightarrow v_5$.\\
\\
\textbf{Case $8d$.} Here there is exactly one root vertex ($v_1$), one internal split vertex ($v_2$), one internal
recombination vertex ($v_3$) and one external recombination vertex $(v_4)$. This time the paths to consider, in order,
are $v_1 \rightarrow v_2$, $v_1 \rightarrow v_4$, $v_2 \rightarrow v_3$ (twice)
and $v_3 \rightarrow v_4$.\\
\\
The network has $O(n)$ arcs by Lemma~\ref{lem:arcs} and hence also $O(n)$ vertices. The algorithm makes $O(1)$ calls
to \emph{LeafPathBetween}. Each execution of \emph{LeafPathBetween} inspects at most $O(n)$ vertices, and must check
each vertex against the $O(n)$ vertices in $SEEN$, giving $O(n^2)$ running time.\\
\\
The sides the leaves are on and the order of the leaves on these sides can be determined during the algorithm by
noting that each call to \emph{LeafPathBetween} corresponds to a particular side of the simple level-2 network, and
that each such call explores all leaves hanging off that side.\\
\\
We conclude that the algorithm including construction of look-up tables takes $O(n^2)$ time, and that subsequent
triplet consistency checks in $N$ take $O(1)$ time. $\Box$
\end{proof}

\begin{lemma}\label{lem:8csimple}
The set of triplets $\tilde{T}$ is only consistent with simple level-2 networks.
\end{lemma}
\begin{proof}
Suppose that $N$ is consistent with $\tilde{T}$ but not a simple level-$\leq 2$ network. Then by
Lemma~\ref{lem:withoutcut} $N$ contains a nontrivial cut-arc $a$. Let $B$ be the set of leaves below $a$ and
$A=L\setminus B$. Because $a$ is a nontrivial cut-arc $B$ contains at least two leaves.\\
\\
For every two leaves $x$ and $y$ in $B$ and for every leaf $z$ in $A$ there is only one triplet on these three leaves
that is consistent with the network. Every set of three leaves for which there is only one triplet is a subset of
$\{a,b,c,d,e,f\}$. Hence $x,y,z$ are elements of $\{a,b,c,d,e,f\}$. This holds for any two leaves $x$ and $y$ from $B$
and $z$ from $A$, hence all leaves are elements of $\{a,b,c,d,e,f\}$. This yields a contradiction.\\
\\
It is clear that $\tilde{T}$ is not consistent with a basic tree and not with a simple level-1 network since
$\tilde{T}$ contains three triplets over $\{g,h,a\}$. Hence is $\tilde{T}$ only consistent with simple level-2
networks. $\Box$
\end{proof}

\begin{lemma}\label{lem:8c8c}
The set of triplets $\tilde{T}$ is only consistent with networks of type $8c$ where $g$ and $h$ are recombination
leaves.
\end{lemma}
\begin{proof}
From Lemma~\ref{lem:8csimple} we know that a network $N$ consistent with $\tilde{T}$ must be a simple level-2 network.
We first argue that in networks of type $8b$ and $8c$ $g$ and $h$ have to be recombination leaves. If a leaf $x$ is
not reachable from any recombination vertex then there exists a unique path from the root to $x$. Hence if three
leaves $x$, $y$, $z$ are all not reachable from any recombination vertex then there is only one triplet on $\{x,y,z\}$
consistent with the network. It follows that if for any three leaves there are two triplets in the input (a
\emph{double triplet}) then at least one of these leaves must be reachable from a recombination vertex. In $8b$ and
$8c$ there are precisely two leaves that are reachable from a recombination vertex. Therefore, if the network is of
type $8b$ or $8c$ then it is clear that $g$ and $h$ have to be recombination leaves, since they are the only two
leaves
that together appear in all double triplets.\\
\\
First consider networks of type $8b$ and observe that if for any three leaves there are three triplets in the input (a
\emph{triple triplet}), then this input can only be consistent with a network of type $8b$ if these leaves are on
sides $G$, $H$ and $C$. Because $\{g,h,x\}$ is a triple triplet for every $x \neq g, h$, all leaves but $g$ and $h$
have to be on side $C$. But in this case it is not possible for both triplets $eg|f$ and $fg|e$ to be simultaneously
consistent
with the network, since $g$ is on side $G$ or $H$ and $e$ and $f$ are both on side $C$.\\
\\
Now consider a network of type $8a$ and observe that this network can only be consistent with triple triplets if its
leaves are on sides $C$, $E$ and $F$ or on sides $C$, $D$ and $F$. In the input is a triple triplet $\{g,h,x\}$ for
all $x \neq g, h$. From this it follows that there are only two possibilities for the network to look like. The first
possibility is that $g$ and $h$ are on the sides $F$ and $D$ or the sides $F$ and $E$ and all other leaves are on side
$C$. But in this case the triplets $eg|f$ and $fg|e$ cannot simultaneously be consistent with the network, since $e$
and $f$ are both on side $C$. The other possibility is that $g$ and $h$ are on the sides $F$ and $C$ and all other
leaves are on the sides $D$ and $E$. From the triplet $ec|g$ it follows that $e$ and $c$ are on the same side. But in
that case $ge|c$ and $hc|e$ cannot simultaneously be consistent with the network.\\
\\
Finally, consider networks of type $8d$. The only way for a triple triplet to be consistent with this type of network
is to put the leaves in the triple triplet on the sides $B$, $C$ and $F$. Since $g$ and $h$ are in a triple triplet
with every other leaf we know that $g$ and $h$ are on the sides $F$ and (without loss of generality) $B$ and all other
leaves are on side $C$. But in this case it is not possible that triplets $eg|f$ and $fg|e$ are simultaneously
consistent with the network, since $e$ and $f$ are both on side $C$. $\Box$
\end{proof}

\begin{lemma}\label{lem:8cunique}
The set of triplets $\tilde{T}$ is only consistent with $\tilde{N}$.
\end{lemma}
\begin{proof}
Let $N$ be a network consistent with $\tilde{T}$. From Lemma~\ref{lem:8c8c} we know that $N$ is of type $8c$ and that
$g$ and $h$ are the two recombination leaves. Since there is no triplet $ab|g$ we know that $a$ and $b$ are on
different sides (one on the left and one on the right side). Assume without loss of generality that $a$ is on side
$A$, $C$ or $E$, $b$ is on side $B$, $D$ or $F$, $g$ is on side $G$ and
$h$ on side $H$.\\
\\
From the triplets $ac|g$ and $ae|g$ it follows that $c$ and $e$ are both on one of the sides $A$, $C$ or $E$. And from
the triplets $bd|g$ and $bf|g$ it follows that $d$ and $f$ are both on one of the sides $B$, $D$ or $F$.\\
\\
From the triplets $ch|e$ and $eg|c$ it now follows that $c$ is on side $C$ and $e$ on side $E$. And from the triplet
$ce|a$ then follows that $a$ is on side $A$. Similarly, from the triplets $dh|f$ and $fg|d$ it follows that $d$ is on
side $D$ and $f$ on side $F$. And from the triplet $df|b$ then follows that $b$ is on side $B$. Therefore, $N =
\tilde{N}$. $\Box$
\end{proof}

\end{document}